\documentclass[aps,reprint,superscriptaddress,unsortedaddress,amsmath,amssymb]{revtex4-2}
\usepackage[utf8]{inputenc}
\usepackage{amsfonts}
\usepackage{times}
\usepackage[cmyk]{xcolor}
\usepackage{tikz}
\usetikzlibrary{shapes.geometric}
\usepackage{xcolor}
\usepackage{verbatim}
\usepackage{pgfplotstable}
\usepackage{url}
\usepackage[hidelinks]{hyperref}
\usepackage{subcaption}
\usepackage{amsthm}
\usepackage{amsmath}
\usepackage{tabularx}
\usepackage{graphicx}
\usepackage{array}
\usepackage{algorithm}
\usepackage{algpseudocode}
\usepackage{multirow}

\usetikzlibrary{angles}
\usetikzlibrary{quantikz, patterns}
\usepackage{cleveref}
\newtheorem{thm}{Theorem}\crefname{thm}{Theorem}{Theorems}
\newtheorem{lem}{Lemma}\crefname{lem}{Lemma}{Lemmas}
\newtheorem{prp}{Proposition}\crefname{prp}{Proposition}{Propositions}
\crefname{cor}{Corollary}{Corollaries}
\crefname{dfn}{Definition}{Definitions}
\crefname{section}{Section}{Sections}
\crefname{appendix}{Appendix}{Appendices}

\newcommand{\SQiSW}{\mathrm{SQ{i}SW} }
\newcommand{\iSWAP}{\mathrm{{i}SWAP} }
\newcommand{\CNOT}{\mathrm{CNOT} }
\newcommand{\CPhase}{\mathrm{CPHASE} }
\newcommand{\SWAP}{\mathrm{SWAP} }

\date{\today}

\begin{document}
\title{Compiling Arbitrary Single-Qubit Gates Via the Phase-Shifts of Microwave Pulses}
\author{Jianxin Chen}
\author{Dawei Ding}
\author{Cupjin Huang}
\affiliation{Alibaba Quantum Laboratory, Alibaba Group USA, Bellevue, Washington 98004, USA}
\author{Qi Ye}
\affiliation{Alibaba Quantum Laboratory, Alibaba Group, Hangzhou, Zhejiang 311121, People's Republic of China
\\
Institute for Interdisciplinary Information Sciences, Tsinghua University, Beijing 100084, People's Republic of China}

\begin{abstract}
We give an arbitrary single-qubit gate compilation scheme on superconducting processors that takes advantage of tuning the phase shift of microwave pulses to obtain a continuous gate set. This scheme is compatible with any two-qubit gate, and we only need to calibrate the $X_\pi$ and $X_{\pi/ 2}$ pulses. We implement this on fluxonium and obtain state-of-the-art fidelities. We give two other schemes: the first requires one $X_\pi$ pulse and one pulse with a variable rotation angle, and the second requires four $X_{\pi/2}$ pulses. We also find that if we can do virtual $Z$ gates, then we can also do virtual gates around any axis. Our results apply to any physical platform that natively supports virtual $Z$.
\end{abstract}

\maketitle

\section{Introduction}
\label{sec:pmw}

The conventional way to realize universal quantum computing is to compile a quantum circuit using arbitrary single-qubit gates, that is, all of $SU(2)$, and a fixed two-qubit gate, such as $\CNOT$. Hence, arbitrary single-qubit gates are ultimately necessary to run quantum algorithms. Furthermore, arbitrary single-qubit gates can help alleviate two-qubit gate errors~\cite{lao2022software} and are necessary for certain benchmarking schemes~\cite{huang2021towards,kong2021framework}. However, it is a nontrivial task to realize arbitrary single-qubit gates on an experimental device. One way to realize arbitrary single-qubit gates is via approximate compilation using a finite universal single-qubit gate set, such as $\{H, S, T\}$~\cite{nielsen2010quantum}. However, the gate fidelities on current experimental devices are limited, with decoherence still being a significant issue. This approach would therefore lead to unacceptably low fidelities. 

Another way to realize arbitrary single-qubit gates is via exact compilation using continuous gate sets, such as the Euler angle decomposition:
\begin{align}
\label{eq:zxz}
    U = Z_\theta X_\sigma Z_\phi,
\end{align}
where $X_\sigma := e^{-i\sigma X/2}$ and $Z_\phi := e^{-i \phi Z/2}$.
The problem with this, however, is that we need to physically implement infinitely many single-qubit gates. This is difficult because the physical control parameters on an experimental device can have complicated relationships with the parameters of the abstract unitary gate being implemented. For example, for a superconducting qubit, which is the physical platform we will focus on, single-qubit gates are usually implemented via a microwave pulse sent on a drive line capacitively coupled to the qubit. For illustration, the circuit diagram for a Cooper pair box qubit is shown in~\Cref{fig:mw_qubit}.
\begin{figure}[h]
	\centering
	\includegraphics[width =0.4\textwidth]{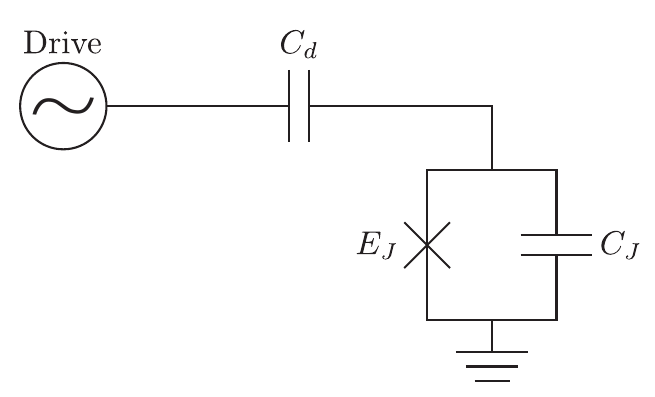}
	\caption{A microwave drive line capacitively coupled to a Cooper pair box. }
	\label{fig:mw_qubit}
\end{figure}
This pulse implements an $X_\sigma$ gate, where the rotation angle $\sigma$ is determined by the amplitude and duration of the pulse. However, there are complicating factors: to implement the desired rotation angle to high precision we would also need to take into account the pulse shape, as well as pulse distortion at short time scales, and even complex processes such as leakage and phase errors induced by the AC-Stark shift~\cite{chen2016measuring}. In experiments, it is therefore necessary to calibrate the gate parameters such as duration, amplitude, and detuning, for every rotation angle we want to implement (See for instance~\cite{barends2014rolling}.).

There is a known way to accurately implement a continuous gate set without an infinite calibration overhead for superconducting qubits. It is well-known that in the rotating frame, by adding a phase shift $\phi$ to the microwave pulse implementing an $X_\sigma$ gate, in the frame rotating with the qubit we instead obtain the gate~\cite{krantz2019quantum}
\begin{align}
    \label{eq:conj_x_rot}
    X_\sigma(\phi)\equiv Z_{-\phi}X_\sigma Z_\phi.
\end{align}
We call $\sigma$ the rotation angle of $X_\sigma(\phi)$, and below we will often fix $\sigma$ and take $\phi$ to be a free variable. This is very exciting observation; by calibrating an $X_\sigma$ gate, we get a continuous gate set $\{X_\sigma(\phi)\}_\phi$ by tuning $\phi$. Furthermore, the correspondence between the phase shift $\phi$ and the conjugation angle $\phi$ in $Z_{-\phi}X_\sigma Z_\phi$ is not complicated by other factors. That is, to high accuracy, the phase shift $\phi$ is equal to the conjugation angle $\phi$~\cite{mckay2017efficient}. Furthermore, the phase shift $\phi$ can made be very precise experimentally~\cite{mckay2017efficient} using a global frequency reference such as an atomic clock~\cite{ball2016role}. We will refer to this trick as phase-shifted microwave (PMW) pulses.

We can take advantage of this continuous gate set to compile arbitrary single-qubit gates. This was already done in~\cite{mckay2017efficient}, arbitrary single-qubit gates are compiled using 
\begin{align}
\label{eq:vz_scheme}
    Z_{\theta}X_{\pi/2}Z_{\phi}X_{\pi/2}Z_{\omega}.
\end{align}
This can be compiled via two $X_\sigma(\phi)$ rotations, save for an extraneous $Z$ rotation $Z_{-(\theta+\phi+\omega)}$. We will refer to this as the~\emph{virtual $Z$ scheme}, as the rightmost $Z_\phi$ in~\Cref{eq:conj_x_rot} is effectively used as a $Z$ rotation. 
The extraneous $Z$ rotation can be taken care of with classical post-processing that tracks all the phases in a procedure known as \emph{phase carrying}, which usually assumes that the two-qubit gate used has a special matrix structure such that the extraneous phases can be carried across it. We call such a gate a \emph{phase carrier}. For example, for an $\iSWAP$ gate, the two extraneous phases can be carried across the gate, but the phases will be exchanged between the two qubits. Another common phase carrier is the $\mathrm{CZ}$ gate. For a detailed explanation of phase carrying as well as a mathematical characterization of these special two-qubit gates and their local equivalence classes, see~\Cref{sec:phase_carrier}.
 
The assumption that the two-qubit gate being used is a phase carrier does not always hold. Although microwave-activated gates can get around this, the usual approach for realizing gates such as $\SQiSW \equiv \sqrt{\iSWAP}$ or other general $\iSWAP$ family gates via flux tuning is explicitly incompatible with phase carrying. Note that this incompatibility is also with the use of the virtual $Z$ gate to realize arbitrary $Z$ rotations as corrections to two-qubit gates~\cite{krantz2019quantum,mckay2017efficient,mi2021information,qiskitRZ}. Other examples of incompatible two-qubit gate schemes include the dynamical decoupled two-qubit gate (we refer to this as $\mathrm{DDCZ}$) in~\cite{guo2018dephasing} and gates realized through tunable couplers, such as general $\mathrm{fSim}$ family~\cite{foxen2020demonstrating} and $\mathrm{bSWAP}$ family~\cite{roth2017analysis} gates. Furthermore, new gate schemes are still being being developed, and compatibility with the virtual $Z$ scheme can be a quite restrictive requirement. 

In this letter we propose novel compilation schemes that still take advantage of using PMW pulses to exactly compile arbitrary single-qubit gates but without incurring extraneous $Z$ rotations~\footnote{To disambiguate, we refer to virtual $Z$ gates or schemes as the use of PMW pulses to obtain an effective $Z$ rotation $\cdot Z_\phi$, while we refer to our proposed PMW gates or schemes as the use PMW pulses to obtain an effective $Z$ conjugation $Z_{-\phi} \cdot Z_\phi$. The former requires phase carrying while the latter does not. In other words, PMW compiled gates do not contribute to tracked phases. }. The resulting compilation schemes are compatible with any two-qubit gate. In particular, one of the schemes requires just three $X_\sigma(\phi)$ gates with fixed rotation angles. This is the minimal necessary since $SU(2)$ has 3 real dimensions. It requires only the calibration of $X_\pi$ and $X_{\pi/2}$ pulses, which are already necessary for Clifford-based randomized benchmarking and the measurement of $T_1$ and $T_2$ coherence times. We additionally provide other variants of this scheme, one that requires two $X_\sigma(\phi)$ gates, but one of which with a variable rotation angle. The other scheme requires four $X_\sigma(\phi)$ gates but with the benefit that only the $X_{\pi/2}$ pulse needs to be calibrated. Lastly, we prove that the ability to do virtual $Z$ implies we can do single-qubit virtual rotations \emph{around any axis}, not just $Z$. We give compilation schemes based on this fact that are not only compatible with phase carriers but also with gates locally equivalent to one, such as $\CNOT$, the cross-resonance ($\mathrm{CR}$) gate~\footnote{Note however that the cross-resonance gate is microwave-activated~\cite{paraoanu2006microwave,rigetti2010fully}, in which case we do not need the two-qubit gate to be a phase carrier. }, or even the M\o lmer-S\o rensen ($\mathrm{MS}$) gate\cite{sorensen2000entanglement} on trapped-ion systems, all of which are locally equivalent to $\mathrm{CZ}$. Note that since we treat everything at an abstract level, our compilation schemes can be used on any quantum computing platform which natively supports $X_\sigma(\phi)$ gates. This includes trapped-ion systems~\cite{bruzewicz2019trapped,saki2021shuttle}, quantum dots~\cite{burkard2021semiconductor}, and nitrogen-vacancy centers~\cite{dobrovitski2013quantum}.

We summarize all single-qubit gate compilation schemes that use PMW pulses in~\Cref{tab:results}, along with the pulses needed for each and their two-qubit gate compatibilities. We also list examples of native two-qubit gates for which these compilation schemes would be suitable for. Note this only takes into account the matrix form and not the physical implementation. 
\begin{table}[hpt]
    \centering
    \begin{tabular}{|c|c|c|c|}
    \hline
      Scheme   & Pulses & 2Q Compatibility & Example Gates\\
      \hline
        vZ-1 &  1V & \multirow{2}{*}{Phase Carrier}& $\mathrm{CPHASE}(\theta)$\\
        \cline{1-2}
        vZ-2 & 2F & & $\iSWAP$, $ \mathrm{bSWAP}$\\\hline
        PMW-3  & 3F & \multirow{3}{*}{All-Compatible} & $\SQiSW$, $\iSWAP(\theta)$\\
        \cline{1-2}
        PMW-4 &  4F & & $\mathrm{bSWAP}(\theta)$\\
        \cline{1-2}
        PMW-2 & 1V+1F & & $\mathrm{fSim}(\theta,\phi)$\\
        \hline
        vR-1 & 1V & \multirow{2}{*}{Leaky Gate} & $\CNOT$, $\mathrm{CR}$\\
        \cline{1-2}
        vR-2 &  2F & & $\mathrm{DDCZ}$, $\mathrm{MS}$ \\
        \hline
    \end{tabular}
    \caption{List of the compilation schemes in this letter. In the ``Pulses'' column, the different letters stand for fixed (F) rotation angle and variable (V) rotation angle. The ``Example Gates'' column lists native two-qubit gates for which the scheme would be suitable for.}
    \label{tab:results}
\end{table}
Other than the vZ schemes~\cite{mckay2017efficient}, all compilation schemes are novel and are elaborated on in~\Cref{sec:pmw}. We conclude in~\Cref{sec:discussion} with a discussion.

\section{PMW Compilation schemes}
\label{sec:pmw}
In this section we give the different all-compatible compilation schemes based on PMW pulses. For completeness, we first provide the standard exposition on single-qubit gates via microwave pulses starting with the circuit Hamiltonian~\cite{krantz2019quantum}. In~\Cref{fig:mw_qubit} we have a standard example of a microwave drive line coupled to a superconducting qubit circuit, this one being the Cooper pair box. Note that a resonant pulse on the flux line would lead to the same qubit Hamiltonian, and so our PMW compilation schemes would apply in that case as well.

The Hamiltonian of such a system projected onto the computational subspace (spanned by ground state $\vert 0 \rangle$ and first excited state $\vert 1 \rangle$) is given by
\begin{align*}
	H(t) = -\frac{\omega_{01}}{2} Z + \Omega(t) \cos(\omega_D t+ \phi) X,
\end{align*}
where $\hbar \omega_{01} \equiv E_1 - E_0$ is the energy difference between $\vert 1\rangle$ and $\vert 0\rangle$, $\Omega(t)$ is the envelope of the microwave pulse, $\omega_D$ its frequency, and $\phi$ its initial phase. Then, in the frame rotating with the qubit, when the drive is resonant the Hamiltonian becomes
\begin{align*}
	H_\text{rf}(t) = \frac{\Omega(t)}{2} (\cos \phi X - \sin \phi Y).
\end{align*}
Note that here we applied the rotating wave approximation. When implemented for a time $T$, the gate effected is
\begin{align*}
	U(T)& = \exp\left[ -\frac{i}{2} (\cos \phi X - \sin \phi Y) \int_0^T \Omega(t) dt \right] \\
	&= X_\sigma(\phi),
\end{align*}
where $\sigma\equiv \int_0^T \Omega(t) dt$. We obtain the desired $X_\sigma(\phi)$ rotation. Hence, by varying the initial phase $\phi$ of the microwave pulse, we have a direct handle on the parameter of the unitary we want to implement.

\subsection{Three fixed angle rotations (PMW-3)}
The first scheme we provide can compile an arbitrary single-qubit gate using three $X_\sigma(\phi)$ rotations with fixed rotation angles. An arbitrary single-qubit gate $U \in SU(2)$ can be parameterized by three real parameters:
\begin{align}
\label{eq:1q_gate}
    U(\alpha, \beta, \gamma) = \begin{bmatrix}
        e^{i\alpha}\cos \gamma & - e^{-i\beta} \sin\gamma\\
        e^{i\beta}\sin \gamma & e^{-i\alpha}\cos \gamma
    \end{bmatrix}.
\end{align}
A direct calculation yields that multiplying three $X_\sigma(\phi)$ rotations with rotation angles $\pi/2, \pi,\pi/2$ gives
\begin{align}
    &X_{\pi/2}(\theta)X_\pi(\phi) X_{\pi/2}(\omega) \nonumber\\
    &= 
\begin{bmatrix}
-e^{i(\theta-\omega)/2} \cos\frac{\theta-2\phi+\omega}{2} & -e^{i(\theta+\omega)/2}\sin \frac{\theta-2\phi+\omega}{2}\\
e^{-i(\theta+\omega)/2}\sin \frac{\theta-2\phi+\omega}{2} & -e^{-i(\theta-\omega)/2} \cos\frac{\theta-2\phi+\omega}{2}
\label{eq:3_pmw}
\end{bmatrix}.
\end{align}
We can therefore identify~\Cref{eq:3_pmw} with~\Cref{eq:1q_gate} by setting
\begin{align*}
    (\theta-2\phi+\omega)/2 & = \pi- \gamma\\
    (\theta-\omega)/2 &= \alpha\\
    (\theta+\omega)/2 &= -\beta,
\end{align*}
which is solved by
\begin{align}
\label{eq:3_pmw_angles}
    \theta = \alpha - \beta, \quad \omega  = -\alpha-\beta, \quad \phi = -\beta +\gamma -\pi.
\end{align}
We will refer to this as the \emph{PMW-3 scheme}.

The variable parameters $\alpha, \beta, \gamma$ are directly related to the phase shifts of the microwave pulses $\theta, \phi, \omega$ according to~\Cref{eq:3_pmw_angles}. However, unlike the virtual $Z$ scheme, this scheme does not incur an extraneous $Z$ rotation. Furthermore, only the pulses for $X_{\pi/2}$ and $X_\pi$ need to be calibrated, which are already necessary for realizing Clifford-based randomized benchmarking and $T_1,T_2$ measurements~\cite{krantz2019quantum}. Now, although we need three $X_\sigma(\phi)$ rotations to compile an arbitrary single-qubit gate, certain single-qubit gates can be compiled using fewer rotations, as we show in~\Cref{sec:special_1q}. We also show in~\Cref{sec:enc_gates} that for excitation number conserving two-qubit gates, we can reduce the number of pulses needed for a pair of single-qubit gates adjacent to the two-qubit gate from 6 pulses to 5. Lastly, it is not a coincidence that we come upon this PMW-3 scheme; in~\Cref{sec:uniqueness} we surprisingly find that any compilation scheme for arbitrary single-qubit gates using the minimial number of three $X_\sigma(\phi)$ rotations is essentially the same as our scheme. 

In \cite{wang2022realizing}, the authors study the two-qubit gate $\SQiSW$. To benchmark this gate, they use fully randomized benchmarking (FRB), which requires Haar random gates. Each Haar random gate is compiled using $\SQiSW$ and arbitrary single-qubit gates. The single-qubit gates are compiled using PMW-3. Furthermore, through various calibration techniques, they can measure the single-qubit gate errors on their implementation of $\SQiSW$, and they correct for this at compilation time. For the arbitrary single-qubit gates compiled using PMW-3, fidelities per pulse of more than 99.9\% are achieved, which is comparable with the state-of-the-art.

\subsection{Four fixed angle rotations (PMW-4)}
Since any single-qubit gate has three real parameters, we need at least three phase shifts to cover all the degrees of freedom. Interestingly, by using one more pulse we can further alleviate the burden of calibration but at the cost of incurring a greater overall error. More specifically, the PMW-3 scheme only requires calibrating two fixed $X$ rotation angles: $\pi$ and $\pi/2$. We can reduce this to a single angle $\pi/2$ by splitting the $\pi$ rotation into two $\pi/2$ rotations. To be specific, we propose a \emph{PMW-4 scheme}:
\begin{align}
    &U(\alpha, \beta, \gamma)\nonumber\\
    &= X_{\pi/2}(\theta) X_{\pi/2}(\phi_1)X_{\pi/2}(\phi_2)X_{\pi/2}(\omega)
    .
\label{eq:4_pmw}
\end{align}
Since $X_{\pi/2}X_{\pi/2}=X_{\pi}$, we can verify from~\Cref{eq:3_pmw_angles} that a solution to~\Cref{eq:4_pmw} is
\begin{align*}
    \theta = \alpha - \beta, \quad \omega  = -\alpha-\beta, \quad \phi_1 = \phi_2 = -\beta +\gamma -\pi.
\end{align*}

\subsection{One variable angle and one fixed angle rotation (PMW-2)}
We also give a compilation scheme involving one variable rotation angle and one fixed rotation angle $\pi$. In particular, we can do a direct calculation to conclude
\begin{align*}
    X_\sigma(\theta)X_\pi(\omega) = U(\alpha, \beta, \gamma)
\end{align*}
by setting
\begin{align*}
    \sigma = 2 \gamma - \pi, \quad \theta = 3 \pi/2+\alpha-\beta, \quad \omega = 3 \pi/2-\beta.
\end{align*}
Although this involves a $X_\sigma(\phi)$ rotation with a variable rotation angle and therefore would have calibration issues, we only need two rotations instead of three as in the PMW-3 scheme. We will refer to this as the \emph{PMW-2 scheme}. This scheme is to the $ZXZ$ scheme via~\Cref{eq:zxz} using a $X_\sigma(\phi)$ gate with an extraneous phase as the PMW-3 scheme is to the virtual $Z$ scheme. Since it involves a $X_\pi(\phi)$ rotation, additional $Z$ rotations can be incorporated in a similar manner to that of the PMW-3 scheme.

\subsection{Compatibility with leaky gates (virtual $R$)}
Here we generalize the virtual $Z$ scheme to make it compatible with any leaky gate. A two-qubit gate $V$ is called \emph{leaky}~\cite{peterson2020fixed} on the first wire, if there exists an arbitrary axis single-qubit rotation $R_{\phi}$ and single-qubit gate families $C(\phi), D(\phi)$ such that
$$V(R_\phi\otimes I)=(C(\phi)\otimes D(\phi))V.$$
Leaky gates are a direct generalization of phase carriers, and there are many experimentally relevant gates that are leaky but are not phase carriers, such as the $\mathrm{DDCZ}$, the $\CNOT$ gate and the  or the M\o lmer-S\o rensen gates. In these cases the virtual $Z$ compilation scheme is not directly applicable. Instead, we propose virtual $R$ compilation schemes that decomposes a single-qubit gate into a pulse and an extraneous rotation along a given axis. The freedom to choose this axis implies we can perform virtual rotations around any axis. In particular we prove the following: given an arbitrary rotation axis $R=aX+bY+cZ$, for arbitrary $U\in SU(2)$, there exists angles $\phi,\omega, \theta$ such that
$$U = R_\theta X_{\pi/2}(\phi)X_{\pi/2}(\omega).$$

The resulting compilation scheme uses two phase-shifted $X_{\pi/2}$ pulses, just like the virtual $Z$ scheme in~\Cref{eq:vz_scheme}. We can also obtain a generalization of the compilation scheme via~\Cref{eq:zxz} using one $X_\sigma(\phi)$ gate with an extraneous phase. To be explicit, we call these schemes \emph{vR-2} and \emph{vR-1} respectively, and to parallel our vR scheme names, we call the corresponding schemes \emph{vZ-2} and \emph{vZ-1}. The number of pulses is the same as the virtual $Z$ schemes, saving one pulse per gate compared to the PMW compilation schemes. The proof that we can perform virtual $R$ compilation and the corresponding generalization of phase carrying can be found in~\Cref{sec:vr} and~\Cref{sec:leaky} respectively.

\section{Discussion}
\label{sec:discussion}
Using our PMW-3 scheme we can compile single-qubit gates where the gate's three real parameters are directly related to the phase shifts $\phi$ of the microwave pulses that implements the gate $X_\sigma(\phi)\equiv Z_{-\phi} X_\sigma Z_{\phi}$. Experimentally these phase shifts can be made extremely precise, unlike the angle of rotation $\sigma$, which \emph{a priori} needs to be calibrated for each possible value~\cite{barends2014rolling}. The PMW-3 scheme in comparison only requires calibrating the $X_\pi$ and $X_{\pi/2}$ rotations. The PMW-4 scheme only requires calibrating $X_{\pi/2}$, albeit at the cost of adding another pulse. The PMW-2 scheme only requires two pulses, but requires calibrating an arbitrary $X_\sigma$ rotation. We also gave virtual $R$ schemes that extended the virtual $Z$ schemes to two-qubit leaky gates.


To evaluate our PMW compilation schemes, we need to compare them to the existing alternatives. For phase carriers (or locally equivalent), the compilation scheme of choice is the virtual $Z$($R$) scheme as it only involves two pulses. For two-qubit gates that are not locally equivalent to phase carriers, an alternative to our PMW-3 scheme is the virtual $Z$ scheme plus an arbitrary $Z_\theta$ rotation~\cite{mi2021information}. In general, $Z_\theta$ rotations can be performed via fine flux tuning of the qubit frequency. However, the calibration problem of $X_\sigma$ applies here as well. How the precision compares to that of the microwave phase shift, and whether the additional $X_\pi$ rotation in the PMW-3 scheme under-compensates or over-compensates for the precision difference, is something that would have to be evaluated experimentally. The comparison could be affected by the superconducting circuit parameters and the quality of the control hardware. Note that fixed frequency qubits by definition cannot be tuned and so cannot implement $Z$ rotations in this way. However, on the other hand, the fidelity of gates compiled via PMW-3 only depends the fidelity of the $X_\pi$ and $X_{\pi/2}$ gates, which can be guaranteed via careful calibration. Such a fidelity guarantee is difficult to obtain for uncalibrated $X$ or $Z$ rotations.

It can be argued that the necessity for additional gates, either $Z_\theta$ or $X_\pi(\phi)$ rotations, may suggest a preference for phase carriers in experimental settings. However, phase carriers and even their local equivalence classes are an extremely restricted class of gates, as proved in~\Cref{sec:phase_carrier}. Although most standard two-qubit gates such as $\iSWAP$ and $\mathrm{CZ}$ are phase carriers, there is growing interest in gates such as $\SQiSW$~\cite{mi2021information}, and even continuously parameterized two-qubit gates~\cite{abrams2020implementation,foxen2020demonstrating} for NISQ applications. 

Another question that needs to be addressed is the issue of calibrating arbitrary $X$ rotations. One of the advantages of our PMW-3 scheme is that we only need to calibrate the $X_\pi$ and $X_{\pi/2}$ pulses (only $X_{\pi/2}$ for PMW-4 scheme). However, it is not completely clear if a general $X_\sigma$ rotation really needs individual calibration. In~\cite{barends2014rolling}, a good number of rotation angles were calibrated using the rapid technique in~\cite{kelly2014optimal}, and it was claimed that these calibrated parameters provide an interpolation table for intermediate rotation angles. Although intuitive, we are not aware of such an interpolation method that can achieve fidelities comparable to that of the calibrated angles. The notion of geometric gates~\cite{xu2020experimental} addresses this somewhat by choosing compilation schemes which are innately robust to control parameter errors. However, how much this robustness affects the final gate fidelities of arbitrary single-qubit gates is still relatively unstudied. 

We conclude with a few open questions. Although we found that we can compile all of $SU(2)$ using three $X_\sigma(\phi)$ gates of fixed angles, we do not give a comprehensive theory of the relationship between the fixed angles available and what subsets of $SU(2)$ can be compiled. In particular this could answer what minimal fixed set of calibrated angles would be optimal for a specific algorithm, or for Haar random gates. Another interesting avenue is to find optimal compilation schemes for specific two-qubit gates. As shown in~\Cref{sec:enc_gates}, even gates that are not locally equivalent to phase carriers do not have to cost six pulses for a pair of single-qubit gates. As another example, the cross-resonance gate is not a phase carrier but is still compatible with virtual $Z$ because it is microwave-activated. This allows us perform procedures analogous to adding a phase shift to microwave pulses for single-qubit gates~\cite{mckay2017efficient}. The full implications of this for compilation, along with the fact that it is locally equivalent to a phase carrier and is therefore compatible with virtual $X$ on the second qubit, is worth further study.

\begin{acknowledgments}
We would like to thank Chunqing Deng, Yaoyun Shi, Tenghui Wang, Zhaohui Yang, Gengyan Zhang, and Jun Zhang for insightful discussions and comments on our work. DD would like to thank God for all of His provisions.
\end{acknowledgments}

\appendix
\section{Phase Carrying}
\label{sec:phase_carrier}
In this section we go into more depth about phase carrying and phase carriers. 

When a quantum circuit is expressed in terms of single-qubit gates and a fixed two-qubit gate, if we use the virtual $Z$ scheme to compile single-qubit gates, the extraneous $Z$ rotations can cause a phase error for the two-qubit gate that follows~\footnote{Note that if a single-qubit gate follows we can always absorb it into the single-qubit gate we're compiling. }. This problem can be averted if we assume that the two-qubit gate $U$ is a phase carrier, that is if for all $\theta_0, \theta_1$, there exists $\phi_0, \phi_1$ such that
\begin{align}
    U (Z_{\theta_0}\otimes Z_{\theta_1}) = (Z_{\phi_0}\otimes Z_{\phi_1})U.
\end{align}
With this assumption we can perform phase carrying, a procedure that is is visually demonstrated in~\Cref{fig:2q_carrying}. We call this procedure phase carrying since the extraneous $Z$ rotation is ``carried through'' each subsequent gate.
\begin{figure}[h]
    \centering
    \includegraphics[width=0.4\textwidth]{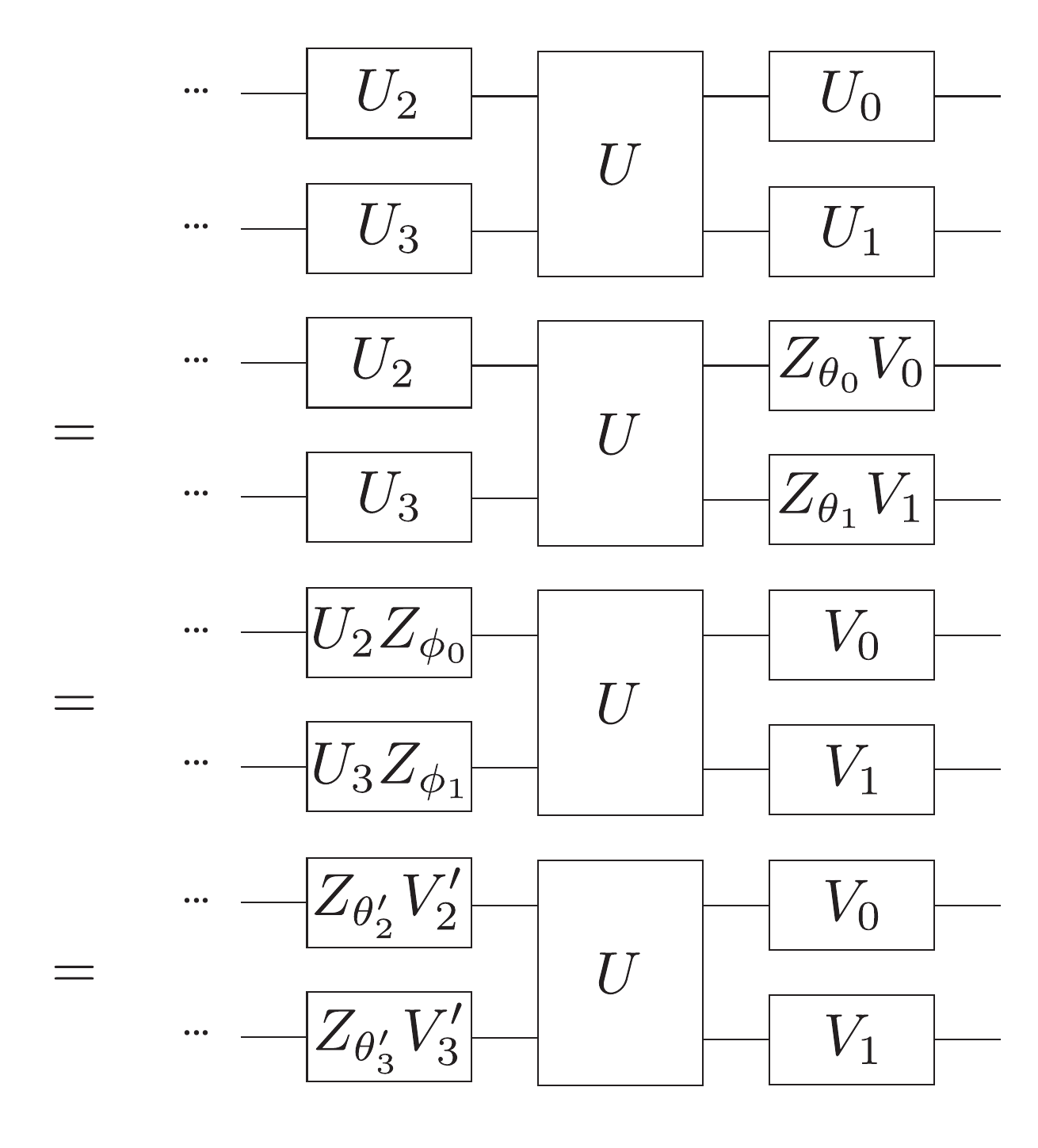}
    \caption{Visualization of compiling an elementary building block of a quantum circuit with arbitrary single-qubit gates and a phase carrier $U$ via the phase carrying procedure. In the first equality, $V_i$ is the part of $U_i$ that can be compiled with two $X_{\pi/2}(\phi)$ gates, with $Z_{\theta_i}$ being the respective extraneous phase. The second equality applies the phase carrier property of $U$. In the last equality, $V_i'$ is the part of $U_i Z_{\phi_j}$ that can be compiled using two $X_{\pi/2}(\phi)$ gates, with $Z_{\theta_i'}$ being the extraneous phase.
    }
    \label{fig:2q_carrying}
\end{figure}

Below we mathematically characterize phase carriers and their local equivalence classes.

\subsection{Proof of phase carrier condition}
We prove in this section the following proposition: \begin{prp}
\label{prp:phase_carrier}
A two-qubit unitary $U \in SU(4)$ is phase a carrier iff the elementwise absolute value of $U$ is a permutation matrix and the permutation $\pi: \mathbb{Z}_2 \times \mathbb{Z}_2 \to \mathbb{Z}_2 \times \mathbb{Z}_2$ satisfies the condition
\begin{align*}
    \forall a,b \in \mathbb{Z}_2, \quad \pi(\overline{a,b}) = \overline{\pi(a,b)},
\end{align*}
where $\overline{(a,b)}$ refers to the bitwise flip operation. 
\end{prp}
Common two-qubit gates such as $\SWAP$, $\mathrm{CZ}$, and $\iSWAP$ are all phase carriers. However, other common two-qubit gates, such as $\CNOT$, the cross-resonance gate~\cite{qiskitRZ}, and $\SQiSW$~\cite{mi2021information}  are not phase carriers. An abstract generalization of phase carriers, called leaky entanglers, was considered in~\cite{peterson2020fixed}. 

\begin{proof}
It will be more convenient to relax the condition
$$U (Z_{\theta_0} \otimes Z_{\theta_1}) = (Z_{\phi_0}\otimes Z_{\phi_1} )U$$
by allowing a global phase. Since $U \in SU(4)$, $Z_\theta \in SU(2)$, there are only four possbilities for the global phase: $\pm 1, \pm i$. For the phase $-1$ we can always set $\phi_0 \to \phi_0 +\pi$ to eliminate it. Hence, the only nontrivial global phases are $\pm i$, which we will show below is impossible. That is, we show that if $U$ satisfies
$$\forall \theta_0, \theta_1, \exists \phi_0, \phi_1 \text{ s.t. } U (Z_{\theta_0} \otimes Z_{\theta_1}) = (Z_{\phi_0}\otimes Z_{\phi_1} )U,$$
up to a global phase, it must satisfy it exactly.
\\~\\
\noindent $(\Rightarrow)$ Suppose $U \in SU(4) $ is a phase carrier up to a global phase. We have
$$Z_\theta \equiv 
\begin{bmatrix}
e^{i \theta/2} & 0 \\
0 & e^{-i\theta/2}
\end{bmatrix} \in SU(2)$$
and
$$Z_\theta \otimes Z_\phi = 
\begin{bmatrix}
e^{i(\theta+\phi)/2} & 0 & 0 & 0\\
0 & e^{i(\theta-\phi)/2} & 0 & 0\\
0 & 0 & e^{i(-\theta+\phi)/2} & 0\\
0 & 0 & 0 & e^{-i(\theta+\phi)/2}\\
\end{bmatrix}.$$
Suppose for contradiction that in row $j$, $U$ has two nonzero entries at $k,k'$. Then, the definition of a phase carrier
$$\forall \theta_0, \theta_1, \exists \phi_0, \phi_1 \text{ s.t. } U (Z_{\theta_0} \otimes Z_{\theta_1}) = (Z_{\phi_0}\otimes Z_{\phi_1} )U$$
up to global phase implies that 
\begin{align*}
(Z_{\theta_0}\otimes Z_{\theta_1})_{kk} =(Z_{\phi_0}\otimes Z_{\phi_1})_{jj} \\
(Z_{\theta_0}\otimes Z_{\theta_1})_{k'k'}=  (Z_{\phi_0}\otimes Z_{\phi_1})_{jj}
\end{align*}
up to a global phase. However, this implies
$$(Z_{\theta_0}\otimes Z_{\theta_1})_{kk} =(Z_{\theta_0}\otimes Z_{\theta_1})_{k'k'}$$
exactly. Looking at the matrix form of $Z_{\theta_0} \otimes Z_{\theta_1}$, it is always possible to choose $\theta_0, \theta_1$ so that
$$(Z_{\theta_0}\otimes Z_{\theta_1})_{kk} \neq (Z_{\theta_0}\otimes Z_{\theta_1})_{k'k'}. $$
Therefore, by contradiction $U$ only has one nonzero entry for each row. Since it is unitary, its elementwise absolute value is a permutation matrix.

Now, denote the permutation $U$ induces by $\pi: \mathbb{Z}_2 \times\mathbb{Z}_2 \to \mathbb{Z}_2 \times\mathbb{Z}_2$. Then, we can identify the elements in $\mathbb{Z}_2 \times \mathbb{Z}_2$ with the entries in $Z_{\theta} \otimes Z_\phi$ as follows:
$$(a,b) \in \mathbb{Z}_2\times \mathbb{Z}_2 \iff e^{i[(-1)^a \theta + (-1)^b \phi]}.$$
With this in mind, it is not hard to see that if the elementwise absolute value of $U$ is a permutation matrix, being a phase carrier up to global phase is equivalent to
\begin{align}
    &\forall \theta_0, \theta_1, \exists \phi_0, \phi_1 \text{ s.t. } e^{i[(-1)^{\pi(a,b)_0} \theta_0 + (-1)^{\pi(a,b)_1} \theta_1]} \nonumber\\
    &=e^{i[(-1)^{a} \phi_0 + (-1)^{b} \phi_1]} \times c
\label{eq:carrier_equiv}
\end{align}
for all $(a,b) \in \mathbb{Z}_2 \times \mathbb{Z}_2$,
where we denote $\pi(a,b)_j$ as the $j$-th element of $\pi(a,b)$ and $c\in\{\pm i\}$. By taking the complex conjugate of both sides, we can conclude
\begin{align*}
&\forall \theta_0,\theta_1 \quad e^{i[(-1)^{\overline{\pi(a,b)}_0} \theta_0 + (-1)^{\overline{\pi(a,b)}_1} \theta_1]} \\
&= e^{i[(-1)^{\pi(\overline{a,b})_0} \theta_0 + (-1)^{\pi(\overline{a,b})_1} \theta_1]} \times c^*/c.
\end{align*}
Setting $\theta_1 =0,$ we have
$$\forall \theta_0 \quad e^{i(-1)^{\overline{\pi(a,b)}_0} \theta_0 } = e^{i(-1)^{\pi(\overline{a,b})_0} \theta_0}  \times c^*/c.$$
Hence, $c^* = c$ and $\pi(\overline{a,b})_0 = \overline{\pi(a,b)}_0$. Similarly, $\pi(\overline{a,b})_1 = \overline{\pi(a,b)}_1$. $c^* = c$ implies there cannot be a nontrivial global phase. The conditions on $\pi$ have also been obtained.
\\~\\
\noindent $(\Leftarrow)$ This follows directly by the equivalent definition of a phase carrier for $U$ whose elementwise absolute value is a permutation matrix in~\Cref{eq:carrier_equiv}, without the global phase. We simply solve the system of equations
$$\phi_0+\phi_1 = (-1)^{\pi(0,0)_0} \theta_0 + (-1)^{\pi(0,0)_1} \theta_1$$
$$\phi_0-\phi_1 = (-1)^{\pi(0,1)_0} \theta_0 + (-1)^{\pi(0,1)_1} \theta_1$$
to obtain the desired $\phi_0, \phi_1$. The conditions on $\pi$ ensures that the other equations
$$-(\phi_0+\phi_1) = (-1)^{\pi(1,1)_0} \theta_0 + (-1)^{\pi(1,1)_1} \theta_1$$
$$-(\phi_0-\phi_1) = (-1)^{\pi(1,0)_0} \theta_0 + (-1)^{\pi(1,0)_1} \theta_1$$
are simultaneously satisfied.
\end{proof}

\subsection{Local equivalence classes of phase carriers}
It is clear the phase carrier condition is not preserved under local unitaries. For example, $\CNOT$ and $\mathrm{CZ}$ are locally equivalent:
\begin{align*}
    \CNOT = (I\otimes H) \mathrm{CZ} (I \otimes H).
\end{align*}
However, the former is not a phase carrier while the latter is. Here we explicitly characterize the local equivalence classes of all phase carrier gates. We will follow the characterization of two-qubit gate local equivalence classes in~\cite{zhang2003geometric}. We give our result as a proposition.
\begin{prp}
\label{prp:local_eq}
The local equivalence classes of phase carrier gates lie on the $\mathrm{I}$ --- $\mathrm{CNOT}$ and $\iSWAP$ --- $\SWAP$ line segments in the Weyl chamber.
\end{prp}
\begin{proof}
By~\Cref{prp:phase_carrier}, the permutation $\pi:\mathbb{Z}_2 \times \mathbb{Z}_2$ a phase carrier corresponds to must satisfy
\begin{align*}
    \forall a,b \in \mathbb{Z}_2 \quad \pi(\overline{a,b}) = \overline{\pi(a,b)}.
\end{align*}
Now, there are four possible choices for $\pi(0,0)$. This then determines $\pi(1,1)$. There are two remaining choices for $\pi(0,1)$, and this determines $\pi(1,0)$. Hence, in total there are eight possibilities for $\pi$. By expressing $\pi$ as a permutation matrix, we can classify the eight possibilities into two cases. In the following $I_2$ refers to the two-dimensional identity gate.
\begin{enumerate}
    \item $\pi = I_2/X \otimes I_2/X$: We can express any phase carrier gate $U$ in this case as 
    \begin{align*}
        U &=
        \begin{bmatrix}
        e^{ia} & 0 & 0 & 0\\
       0 & e^{ib} & 0 & 0\\
       0 & 0 & e^{ic} & 0\\
       0 & 0 & 0 & e^{id}
       \end{bmatrix} \cdot \pi\\
        & =e^{i(b+c)/2 } (Z_{a-c} \otimes Z_{a-b}) \cdot \mathrm{CPHASE}(a-b-c+d)\cdot \pi,
    \end{align*}
    where $e^{ia},e^{ib},e^{ic},e^{id}$ are the matrix elements on the first, second, third, and fourth rows of $U$, respectively~\footnote{Note that technically $U\in SU(4)$, so $d = -(a+b+c)$. }. Hence, $U$ is locally equivalent to $\CPhase(a-b-c+d)$, so its interaction coefficients is of the form $(x,0,0)$ where $x\in[0,\pi/2]$. Geometrically, the local equivalence class lies on the $I$ ---  $\CNOT$ line segment in the Weyl chamber. 
    
    \item $\pi = \SWAP \cdot  (I_2/X \otimes I_2/X)$: 
    We can express any phase carrier gate $U$ in this case as 
    \begin{align*}
        U &=
        \begin{bmatrix}
        e^{ia} & 0 & 0 & 0\\
       0 & e^{ib} & 0 & 0\\
       0 & 0 & e^{ic} & 0\\
       0 & 0 & 0 & e^{id}
       \end{bmatrix} \cdot \pi\\
        & =e^{i(b+c)/2 } (Z_{a-c} \otimes Z_{a-b}) \cdot \mathrm{CPHASE}(a-b-c+d)\\
        & \cdot \SWAP \cdot (I_2/X \otimes I_2/X),
    \end{align*}
    where $e^{ia},e^{ib},e^{ic},e^{id}$ are the matrix elements on the first, second, third, and fourth rows of $U$, respectively. Hence, $U$ is locally equivalent to the mirror gate of  $\CPhase(a-b-c+d)$, so its interaction coefficients is of the form~\cite{cross2019validating}
    \begin{align*}
        (\pi/2,\pi/2,x-\pi/2) &\sim (-\pi/2, \pi/2, x-\pi/2) \\
        &\sim (\pi/2,\pi/2,\pi/2-x),
    \end{align*} 
    where $x\in[0,\pi/2]$. The first equivalence follows by subtracting $\pi$ from the first coordinate and the second equivalence by flipping the signs of the first and third coordinates. Geometrically, the local equivalence class lies on the $\iSWAP$ ---  $\SWAP$ line segment in the Weyl chamber.  
    
\end{enumerate}

\end{proof}

\section{Special single-qubit gates}
\label{sec:special_1q}
We can reduce the number of pulses per single-qubit gate by looking for the following special cases instead of using a rote lookup table such as in~\cite{mckay2017efficient} which will always incur the same number of pulses. We consider below a few special cases when we restrict ourselves to $\pi/2, \pi$ rotation angles.

If we only use $X_\pi(\phi)$ rotations, we can implement single-qubit gates which are either diagonal or anti-diagonal. This result follows from the simple form of a single $X_\pi(\phi)$ rotation:
\begin{align}
\label{eq:conj_x_pi}
    Z_{-\theta} X_\pi Z_{\theta} =
    \begin{bmatrix}
      0 & -i e^{i\theta}\\
      -i e^{-i\theta} & 0
    \end{bmatrix}.
\end{align}

Consider an anti-diagonal single-qubit gate, that is of the form
\begin{align*}
    \begin{bmatrix}
      0 & - e^{-i\beta}\\
      e^{i\theta} & 0
    \end{bmatrix}.
\end{align*}
This equals the expression in~\Cref{eq:conj_x_pi} with setting $\theta = 3\pi/2 -\beta$.

Now consider a diagonal single-qubit gate, that is of the form
\begin{align}
\label{eq:diag_1q}
    \begin{bmatrix}
        e^{i\alpha} & 0 \\
        0 & e^{-i\alpha}
    \end{bmatrix}.
\end{align}
This is a pure $Z$ rotation of angle $-2\alpha$. This can be compiled by two $X_\pi(\phi)$ rotation of opposite phase shifts:
\begin{align}
\label{eq:2_x_pi}
    Z_{\theta} X_\pi Z_{-\theta} Z_{-\theta} X_\pi Z_{\theta} =
        \begin{bmatrix}
            -e^{-2i\theta} & 0 \\
            0 & -e^{2i\theta}
        \end{bmatrix}.
\end{align}
Hence, we can identify~\Cref{eq:diag_1q} with~\Cref{eq:2_x_pi} by setting $\theta = - (\alpha+\pi)/2$.

We next consider the Clifford gates. Interestingly, we can compile any single-qubit Clifford gate using two or less $X_\sigma(\phi)$ rotations. Following~\cite{crooks2020gates}, we can classify single-qubit Clifford gates into three categories:
\begin{enumerate}
    \item Pauli rotations and half-rotations.
\begin{align*}
    I & = I \nonumber\\
    X_{\pi} & = Z_0 X_\pi Z_0 \nonumber\\
    X_{\pi/2} &= Z_0 X_{\pi/2} Z_0\nonumber\\
    X_{-\pi/2} &= Z_{-\pi} X_{\pi/2} Z_\pi\nonumber\\
    Y_\pi & = Z_{\pi/2} X_{\pi} Z_{-\pi/2} \nonumber\\
    Y_{\pi/2} &= Z_{\pi/2} X_{\pi/2} Z_{-\pi/2}\nonumber\\
    Y_{-\pi/2} &= Z_{-\pi/2} X_{\pi/2} Z_{\pi/2} \nonumber\\
    Z_\pi & = (Z_{-\pi/4} X_\pi Z_{\pi/4} )(Z_{\pi/4} X_\pi Z_{-\pi/4}) \nonumber\\
    Z_{\pi/2} & = (Z_{-3\pi/8} X_\pi Z_{3\pi/8} )(Z_{3\pi/8} X_\pi Z_{-3\pi/8}) \nonumber\\
    Z_{-\pi/2} & = (Z_{-5\pi/8} X_\pi Z_{5\pi/8} )(Z_{5\pi/8} X_\pi Z_{-5\pi/8}).
\end{align*}
This explicitly handles all ten gates in this category.

\item Cousins of the Hadamard gate $H$. These are $\pi$ rotations about axes halfway between the $x,y,z$ axes. That is, letting $(n_x,n_y,n_z)$ be a real unit vector, we consider gates of the form
\begin{align*}
    \exp\left[- i \frac{\pi}{2} (n_x X + n_y Y +n_z Z)\right]=
    \begin{bmatrix}
    -i n_z & -n_y -i n_x\\
    n_y-i n_x & i n_z
    \end{bmatrix},
\end{align*}
where one of $n_x, n_y, n_z$ is $0$ and the other two $\pm \frac{1}{\sqrt{2}}$. By setting $n_x = \frac{1}{\sqrt{2}}, n_y = 0, n_z = \frac{1}{\sqrt{2}}$, we obtain $-iH$. We will use the notation $(n_x, n_y, n_z)_\pi$ to refer to the different gates.

When $n_z=0$, the unitary is anti-diagonal. Thus, we can use a single $X_\pi(\phi)$ rotation to compile it. This covers the $(\frac{1}{\sqrt{2}},\frac{1}{\sqrt{2}},0)_\pi, (-\frac{1}{\sqrt{2}},\frac{1}{\sqrt{2}},0)_\pi$ gates.

The other gates in this category ($n_z = \pm \frac{1}{\sqrt{2}}$) can be compiled using one $X_{\pi/2}(\phi)$ rotation and one $X_\pi(\phi)$ rotation:
\begin{align}
    (Z_{-\theta} X_{\pi/2} Z_\theta) (Z_{-\phi} X_{\pi} Z_\phi) = 
    \frac{1}{\sqrt{2}} \begin{bmatrix}
    - e^{i(\theta-\phi)}  & -ie^{i\phi}\\
    -ie^{-i\phi} &-  e^{-i(\theta-\phi)} 
    \end{bmatrix}.
    \label{eq:pi/2-pi}
\end{align}
Thus, we set 
\begin{align}
    e^{i\phi} & = \mathrm{sgn}(n_x) - i \mathrm{sgn}(n_y) \nonumber\\
    e^{i(\theta-\phi)} & = \mathrm{sgn}(n_z),  \label{eq:H_cousins}
\end{align}
where $\mathrm{sgn}:\mathbb{R} \to \mathbb{R}$ is the sign function:
\begin{align*}
    \mathrm{sgn}(x) =
    \begin{cases}
        1 & x > 0\\
        0 & x=0\\
        -1 & x<0.
    \end{cases}
\end{align*}
Since one of $n_x, n_y$ must be $0$,~\Cref{eq:H_cousins} always has a solution. This covers the other four gates in this category.

\item Pauli $Y$ analog of $H$ and its cousins. These are gates of the form
\begin{align*}
    &\exp\left[\mp i \frac{2\pi}{3} (n_x X + n_y Y +n_z Z)\right]\\
    &=
    \frac{1}{2}
    \begin{bmatrix}
    1\mp i \mathrm{sgn}(n_z) & \mp \mathrm{sgn}(n_y)\mp i\mathrm{sgn}(n_x)  \\
    \pm \mathrm{sgn}(n_y)\mp i \mathrm{sgn}(n_x) & 1\pm i \mathrm{sgn}(n_z)
    \end{bmatrix},
\end{align*}
where $n_x, n_y, n_z$ are $\pm \frac{1}{\sqrt{3}}$. Note
\begin{align*}
    \exp\left[i \frac{2\pi}{3} \left(\frac{1}{\sqrt{3}} X + \frac{1}{\sqrt{3}} Y +\frac{1}{\sqrt{3}} Z\right)\right]= 
    \frac{e^{i\frac \pi 4}}{\sqrt{2}}
    \begin{bmatrix}
    1 & 1\\
    i & -i
    \end{bmatrix},
\end{align*}
where the right-hand side is the Pauli $Y$ analog of the Hadamard gate (that is, the columns are the eigenstates of the Pauli $Y$ gate). By comparing with~\Cref{eq:pi/2-pi}, we can compile these gates using one $X_{\pi/2}(\phi)$ and one $X_\pi(\phi)$ rotations via
\begin{align*}
    e^{i\phi} & = \frac{1}{\sqrt{2}}( \pm \mathrm{sgn}(n_x) \mp i \mathrm{sgn}(n_y) )\\
    e^{i(\theta-\phi)} & = \frac{1}{\sqrt{2}} (-1 \pm i \mathrm{sgn}(n_z)).
\end{align*}
Since $n_x,n_y,n_z$ are all nonzero, there is always a solution. This covers all eight gates in this category.
\end{enumerate}
On average, we require $1 \frac{7}{12} \approx 1.58$ pulses per Clifford gate.

\section{Excitation number conserving gates}
\label{sec:enc_gates}
Many two-qubit gates implemented in experiment satisfy a property which we call \emph{excitation number conserving}, which means it preserves the $\mathrm{Span}\{\vert 00\rangle\}$, $\mathrm{Span}\{\vert 01\rangle, \vert 10\rangle\}$, and $\mathrm{Span}\{\vert 11\rangle\}$ subspaces. Equivalently, it satisfies a different version of the phase carrier condition:
\begin{align*}
    \forall \theta, \quad U (Z_\theta \otimes Z_\theta) = (Z_\theta \otimes Z_\theta) U.
\end{align*}
Examples of excitation number conserving gates include $\CPhase$, $\SQiSW$, the $\iSWAP$ family, and the $\mathrm{fSim}$ gate family.

For excitation number conserving gates, we can reduce the necessary $X_\sigma(\phi)$ rotations as compared to a na\"ive application of the PMW-3 scheme via a slight modification of the phase carrying procedure outlined in~\Cref{fig:2q_carrying}. We first compile $U_0$ using the virtual $Z$ scheme, leaving an extraneous rotation $Z_{\theta_0}$. We then compile $Z_{\theta_0} U_1$ using the PMW-3 scheme. This leaves the same extraneous $Z$ rotation on both qubits, and by taking advantage of the excitation number conserving property, we can perform phase carrying. Furthermore, all single-qubit gates directly prior to a measurement can be compiled using the virtual $Z$ scheme. This procedure is visualized in~\Cref{fig:enc_carrying}. Asymptotically, this reduces the number of $X_\sigma(\phi)$ rotations by a factor of $5/6$. Note that this procedure can be trivially extended to more general gates that satisfy: for all $\theta$, there exists $\phi_0, \phi_1$ such that
\begin{align*}
    U (Z_{\theta}\otimes Z_{\theta}) = (Z_{\phi_0}\otimes Z_{\phi_1})U.
\end{align*}

\begin{figure}[h]
    \centering
    \includegraphics[width=0.4\textwidth]{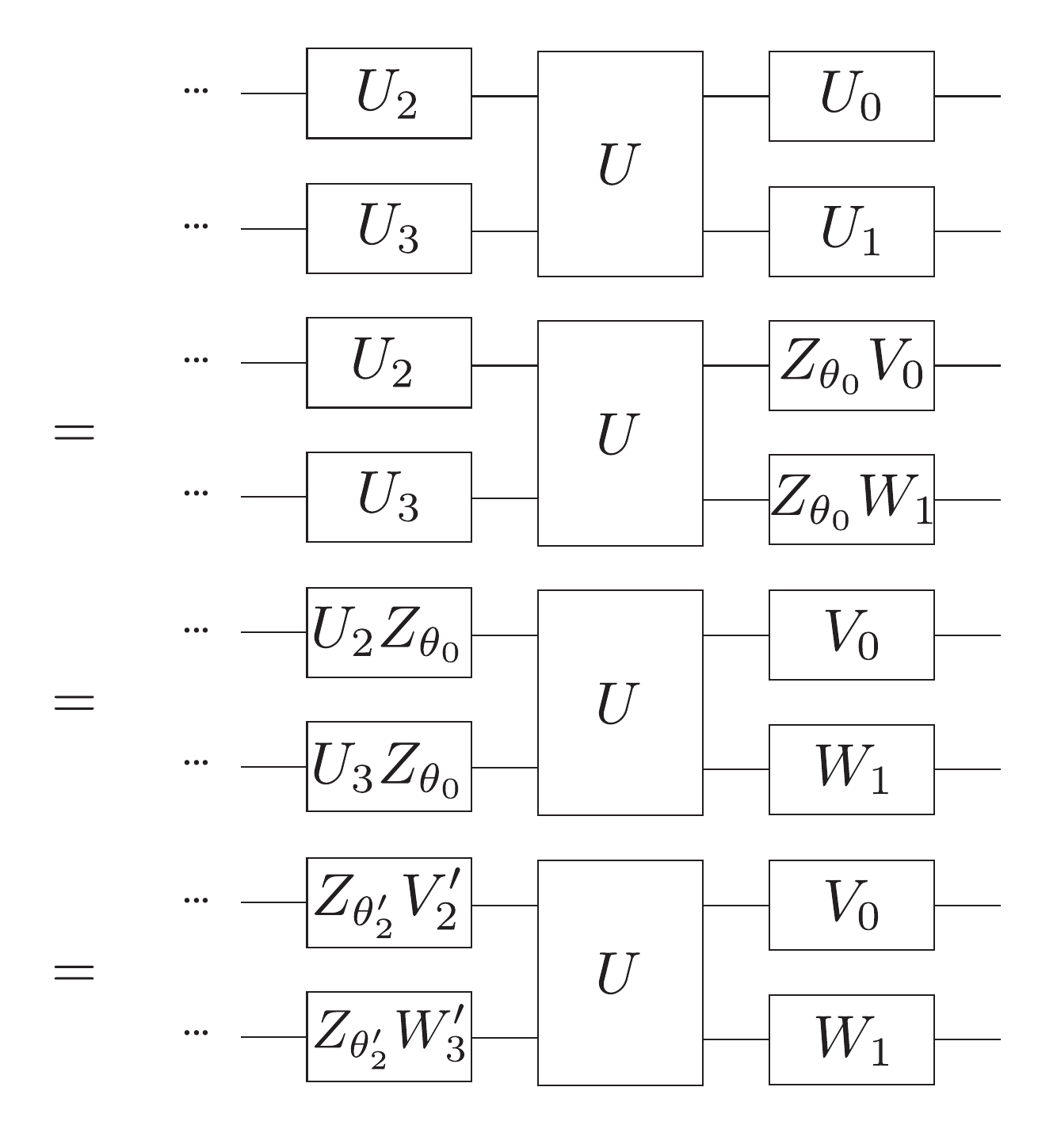}
    \caption{Visualization of compiling an elementary building block of a quantum circuit with a excitation number conserving two-qubit gate $U$. In the first inequality, $V_0, Z_{\theta_0}$ follows form compiling $U_0$ using the virtual $Z$ scheme, and $W_1 \equiv Z_{-\theta_0} U_1$, compiled using PMW-3. The second equality applies the excitation number conserving property of $U$. In the last equality, we compile $U_2 Z_{\theta_0}$ using the virtual $Z$ scheme and $Z_{\theta_2'} U_3 Z_{\theta_0}$ using PMW-3.}
    
    \label{fig:enc_carrying}
\end{figure}

In this modified phase carrying procedure, we often need to compile single-qubit gates plus additional $Z$ rotations. This can also happen when we need to perform phase corrections to the two-qubit gate, as discussed in~\cite{arute2020observation} for $\iSWAP$ family gates. We here note a convenient property of the PMW-3 scheme which makes accommodating additional $Z$ rotations very simple. This stems from the fact that $X_\pi(\phi)$ rotations are unchanged under sandwiching them with the same $Z$ rotation:
\begin{align}
    \label{eq:flip_phase}
    X_\pi(\theta) = Z_{-\theta} X_\pi Z_{\theta} = Z_{\phi -\theta} X_\pi Z_{\phi +\theta}.
\end{align}
Now consider the PMW-3 compilation of a single-qubit gate $U$:
\begin{align*}
    U = (Z_{-\theta} X_{\pi/2} Z_{\theta})(Z_{-\phi} X_{\pi} Z_{\phi})(Z_{-\omega} X_{\pi/2} Z_{\omega}).
\end{align*}
Then, additional $Z$ rotations to $U$ can be absorbed via the following manipulation:
\begin{align*}
    & Z_{\delta_L} U Z_{\delta_R} \\
    &= (Z_{-\theta+\delta_L} X_{\pi/2} Z_{\theta})(Z_{-\phi} X_{\pi} Z_{\phi})(Z_{-\omega} X_{\pi/2} Z_{\omega+\delta_R}) \\
    & = (Z_{-\theta+\delta_L} X_{\pi/2} Z_{\theta-\delta_L})(Z_{-\phi+\delta_L} X_{\pi} Z_{\phi+\delta_R})\\
    &\quad\quad(Z_{-\omega-\delta_R} X_{\pi/2} Z_{\omega+\delta_R}) \\
    & = (Z_{-\theta+\delta_L} X_{\pi/2} Z_{\theta-\delta_L})(Z_{-\phi-\frac{\delta_R-\delta_L}{2}} X_{\pi} Z_{\phi+\frac{\delta_R-\delta_L}{2}})\\
    &\quad\quad(Z_{-\omega-\delta_R} X_{\pi/2} Z_{\omega+\delta_R}),
\end{align*}
where in the last step we use~\Cref{eq:flip_phase}. Hence, by making a simple update
\begin{align*}
    \theta  \mapsto \theta -\delta_L, \quad \phi \mapsto \phi + \frac{\delta_R-\delta_L}{2}, \quad \omega  \mapsto \omega + \delta_R,
\end{align*}
we can compile $U$ plus additional $Z$ rotations.

\section{Uniqueness of the PMW-3 scheme}
\label{sec:uniqueness}
In this section we show that the proposed PMW-3 scheme is essentially unique, in the sense that all feasible schemes have similar forms. For fixed rotation angles $\omega_1,\omega_2,\omega_3$, the PMW-3 scheme $Z_{-\phi_1}X_{\omega_1}Z_{\phi_1}Z_{-\phi_2}X_{\omega_2}Z_{\phi_2}Z_{-\phi_3}X_{\omega_3}Z_{\phi_3}$ can compile arbitrary single-qubit gate if and only if the following set covers $SU(2)/\{\pm I\}$:
\begin{align*}
    &A_{\omega_1,\omega_2,\omega_3}\\
    &=\{Z_{\theta_0}X_{\omega_1}Z_{\theta_1}X_{\omega_2}Z_{\theta_2}X_{\omega_3}Z_{\theta_3}\mid \theta_0+\theta_1+\theta_2+\theta_3=0\}.
\end{align*}
We summarize our uniqueness result with a theorem:
\begin{thm}
\label{thm:uniqueness_PMW-3}
For $\omega_1,\omega_2,\omega_3\in [-\pi,\pi)$, $A_{\omega_1,\omega_2,\omega_3}$ covers $SU(2)/\{\pm I\}$ if and only if one of $\omega_1,\omega_2,\omega_3$ is $\pi$ and the other two are in $\{\pm\pi/2\}$.
\end{thm}

We first introduce a convenient representation of $SU(2)$. It is well known that any element of $SU(2)$ can be uniquely expressed as 
\begin{align*}
    &a_0I-(a_1iZ+a_2iX+a_3iY), \\
    &a_0,a_1,a_2,a_3\in \mathbb{R},\\ &a_0^2+a_1^2+a_2^2+a_3^2=1.
\end{align*}
So there is a natural isomorphism between $SU(2)$ and unit elements in quaternions $\mathbb{H}=\{a_0+a_1i+a_2j+a_3k|a_0,a_1,a_2,a_3\in \mathbb{R}\}$ by mapping $I\mapsto 1, -iZ\mapsto i, -iX\mapsto j, -iY\mapsto k$ (we can verify that $-iZ,-iX,-iY$ satisfy the law of quaternions $ij=-ji=k,jk=-kj=i,ki=-ik=j$).

In this representation, $Z_{\theta_t}=\cos(\theta_t/2)I - \sin(\theta_t/2)iZ$ is written as a unit complex number $z_t=e^{i\theta_t/2}=\cos(\theta_t/2)+\sin(\theta_t/2)i$ and $X_{\omega_t}$ is written as $\cos(\omega_t/2)+\sin(\omega_t/2)j$. Denote $a_t^{(0)}=\cos(\omega_t/2),a_t^{(1)}=\sin(\omega_t/2)$. Using $\overline{z_t}z_t=1$ and $jz_t=\overline{z_t}j$, we can calculate that
\begin{widetext}
\begin{align}
\label{eq:calculation_of_A}
    &Z_{\theta_0}X_{\omega_1}Z_{\theta_1}X_{\omega_2}Z_{\theta_2}X_{\omega_3}Z_{\theta_3} \\
    &= Z_{-\theta_3}Z_{-\theta_2}Z_{-\theta_1}X_{\omega_1}Z_{\theta_1}X_{\omega_2}Z_{\theta_2}X_{\omega_3}Z_{\theta_3}\nonumber\\
    &= \overline{z_3}\overline{z_2}\overline{z_1}(a_1^{(0)}+a_1^{(1)}j)z_1(a_2^{(0)}+a_2^{(1)}j)z_2(a_3^{(0)}+a_3^{(1)}j)z_3\nonumber\\
    &=\left(a_1^{(0)}a_2^{(0)}a_3^{(0)}-a_1^{(0)}a_2^{(1)}a_3^{(1)}u_2-a_1^{(1)}a_2^{(0)}a_3^{(1)}u_1u_2-a_1^{(1)}a_2^{(1)}a_3^{(0)}u_1\right)\nonumber\\
    &+ \left(a_1^{(0)}a_2^{(0)}a_3^{(1)}u_3+a_1^{(0)}a_2^{(1)}a_3^{(0)}u_2u_3+a_1^{(1)}a_2^{(0)}a_3^{(0)}u_1u_2u_3+a_1^{(1)}a_2^{(1)}a_3^{(1)}u_1u_3\right)j\nonumber\\
    &\overset{def}{=} U_{\omega_1,\omega_2,\omega_3}(u_1,u_2,u_3) + V_{\omega_1,\omega_2,\omega_3}(u_1,u_2,u_3)j.
\end{align}
\end{widetext}
Here $u_t=\overline{z_t^2}$ is a unit complex number.  $U_{\omega_1,\omega_2,\omega_3}(u_1,u_2,u_3), V_{\omega_1,\omega_2,\omega_3}(u_1,u_2,u_3)$ are complex functions written in the previous line. The property of unitary matrix ensures 
\begin{align}
\label{eq:norm_1}
    |U_{\omega_1,\omega_2,\omega_3}(u_1,u_2,u_3)|^2+|V_{\omega_1,\omega_2,\omega_3}(u_1,u_2,u_3)|^2=1.
\end{align}

We state a lemma that relates the range of $A_{\omega_1,\omega_2,\omega_3}$ over $SU(2)/\{\pm I\}$ and the range of $U_{\omega_1,\omega_2,\omega_3}$ over the complex numbers. We call the set of  unit complex numbers the unit circle $S^1$ and call the set of complex numbers with norm not larger than unity the closed unit disk $D^2$.

\begin{lem}
$A_{\omega_1,\omega_2,\omega_3}$ covers $SU(2)/\{\pm I\}$ if and only if $\{U_{\omega_1,\omega_2,\omega_3}(u_1,u_2,u_3) \vert u_1,u_2,u_3 \in S^1\}$ covers the set of all antipodal points in $D^2$. 
\end{lem}
\begin{proof}
Per its isomorphism with the quaternions, $SU(2)$ can be written as $\{u+vj|(u,v)\in \mathbb{C}^2,|u|^2+|v|^2=1\}$. So the forward direction is obvious. 

Now we prove the backward direction. Suppose the range of $U_{\omega_1,\omega_2,\omega_3}$ covers all antipodal pairs in $D^2$. For any $u+vj\in SU(2)$, we choose $u_1,u_2,u_3$ such that $u'=U_{\omega_1,\omega_2,\omega_3}(u_1,u_2,u_3)\in\{\pm u\}$. Denote $v'=V_{\omega_1,\omega_2,\omega_3}(u_1,u_2,u_3)$. By~\Cref{eq:calculation_of_A}, we observe that $U_{\omega_1,\omega_2,\omega_3}$ is independent of $u_3$, while all the terms in $V_{\omega_1,\omega_2,\omega_3}$ contain $u_3$. If we change $u_3$ to $u_3e^{i\alpha}$, $(u',v')$ will become $(u',v'e^{i\alpha})$. Since ~\Cref{eq:norm_1} implies $|v'|=|v|$, we can always choose an appropriate $\alpha$ such that $(u',v'e^{i\alpha})\in\{\pm (u,v)\}$. Hence $A_{\omega_1,\omega_2,\omega_3}$ covers $SU(2)/\{\pm I\}$.
\end{proof}

Therefore, all we need to figure out is the range of $U_{\omega_1,\omega_2,\omega_3}$. By~\Cref{eq:calculation_of_A} it has the form 
$$f(u_1,u_2)=b_{00}+b_{01}u_1+b_{10}u_2+b_{11}u_1u_2, $$ 
where $b_{00},b_{01},b_{10},b_{11}\in \mathbb{R}$. The following lemma addresses when such an expression can cover all antipodal pairs in $D^2$.

\begin{lem}\label{lem:range_solution}
Let $R_f:=\bigcup_{u_1,u_2}\{\pm f(u_1,u_2)\}$. $R_f=D^2$ if and only if $b_{00}=0$, one of $b_{01},b_{10},b_{11}$ is $0$ and the other two are $\pm 1/2$.
\end{lem}

Supposing~\Cref{lem:range_solution} is correct, we can apply it to $U_{\omega_1, \omega_2,\omega_3}= a_1^{(0)}a_2^{(0)}a_3^{(0)}-a_1^{(0)}a_2^{(1)}a_3^{(1)}u_2-a_1^{(1)}a_2^{(0)}a_3^{(1)}u_1u_2-a_1^{(1)}a_2^{(1)}a_3^{(0)}u_1$. A simple calculation shows that one of $(a_1^{(0)},a_1^{(1)}),(a_2^{(0)},a_2^{(1)}),(a_3^{(0)},a_3^{(1)})$ is $(0,\pm 1)$ and the other two are $(\pm 1/\sqrt{2},\pm 1/\sqrt{2})$. Hence~\Cref{thm:uniqueness_PMW-3} follows. Now we only need to prove~\Cref{lem:range_solution}.
\begin{proof}
Suppose $b_{11}=0$. Then the ranges of $b_{10}u_1$, $b_{01}u_2$ are circles of radius $|b_{10}|,|b_{01}|$, respectively. And the ranges of $b_{00}+b_{10}u_1+b_{01}u_2$ is the Minkowski sum of the above two circles, thus is a ring centering at $b_{00}$ with outer radius $|b_{10}|+|b_{01}|$ and inner radius $\big||b_{10}|-|b_{01}|\big|$. So the range $R_f$ is a unit disk if and only if $b_{00}=0$, $|b_{10}|=|b_{01}|=1/2$.

Suppose $b_{01}=0$, we can reduce this to the above case by letting $u_1\mapsto u_1\overline{u_2}$ so that $f_2$ becomes $b_{00}+b_{01}u_2+b_{11}u_1$. Hence the only solution is $b_{00}=0, |b_{01}|=|b_{11}|=1/2$. If $b_{10}=0$ the argument is similar.

It suffices to prove that $b_{01},b_{10},b_{11}$ cannot all be nonzero. Assume the opposite. Since $\pm f(e^{i\theta_1},e^{i\theta_2})$ spans the unit disk, we have $\max_{\phi_1,\phi_2\in [-\pi,\pi)}|f(e^{i\phi_1},e^{i\phi_2})|=1$. On the other hand,
\begin{align*}
    &\max_{\phi_1}|f(e^{i\phi_1},e^{i\phi_2})|\\
    =&\max_{\phi_1}|b_{00}+b_{10}e^{i\phi_2}+(b_{01}+b_{11}e^{i\phi_2})e^{i\phi_1}|\\
    =&|b_{00}+b_{10}e^{i\phi_2}|+|b_{01}+b_{11}e^{i\phi_2}|\\
    =&\sqrt{b_{00}^2+b_{10}^2+2b_{00}b_{10}\cos(\phi_2)} \\
    &+ \sqrt{b_{01}^2+b_{11}^2+2b_{01}b_{11}\cos(\phi_2)}.
\end{align*}

Since $b_{01}b_{11}\neq 0$, there are finite number of $\phi_2\in [-\pi,\pi)$ satisfying $\max_{\phi_1}|f(e^{i\phi_1},e^{i\phi_2})|=\max_{\phi_1, \phi_2}|f(e^{i\phi_1},e^{i\phi_2})|=1$. Similarly, there are finite number of $\phi_1\in[-\pi,\pi)$ satisfying $\max_{\phi_2}|f(e^{i\phi_1},e^{i\phi_2})|=1$. It follows that the set $\{(\phi_1,\phi_2)|\phi_1,\phi_2\in [-\pi,\pi),|f(e^{i\phi_1},e^{i\phi_2})|=1\}$ is a finite set and therefore $\pm f(e^{i\phi_1}, e^{i\phi_2})$ cannot cover the whole unit circle.
\end{proof}

\section{Virtual $R$ schemes}
\label{sec:vr}
In this section we prove the following.
\begin{thm}
\label{thm:vr2}
Given an arbitrary rotation axis $R=aX+bY+cZ$, for arbitrary $U\in SU(2)$, there exists angles $\phi,\omega, \theta$ such that
$$U = R_\theta X_{\pi/2}(\phi)X_{\pi/2}(\omega).$$
Alternatively, let $\mathcal{S}=\{X_{\pi/2}(\phi)X_{\pi/2}(\omega)|\phi,\omega\in[0,2\pi]\}$, $\mathcal{T}_R=\{R_\theta|\theta\in[0,2\pi]\}$. Then
$$\forall R, \mathcal{T}_R\cdot \mathcal{S}=SU(2).$$
\end{thm}
This implies that we can perform virtual $R$ gates by compiling $U R_\alpha$ for any $\alpha$ using~\cref{thm:vr2}, obtaining an extraneous $R$ rotation that is carried.
\begin{proof}
We first reduce the problem to a canonical form. Since $R$ and $Z$ have the same spectrum $\{\pm 1\}$, there exists a unitary $U\in SU(2)$ such that $R=UZU^\dag$, hence $\mathcal{T}_R=U\mathcal{T}U^\dag$ for $\mathcal{T}:=\mathcal{T}_Z$. It then suffices to prove that for all $U\in SU(2)$, $$\mathcal{S}\cdot U\cdot \mathcal{T}\cdot U^\dagger=SU(2)\Leftrightarrow \mathcal{S}\cdot U\cdot \mathcal{T}=SU(2)\cdot U=SU(2).$$
This can be further simplified by decomposing $U$ into $Z_\alpha X_\beta Z_\gamma$. Since both $\mathcal{T}$ and $\mathcal{S}$ commutes with $Z$-rotations, the two $Z$-rotations can be commuted out and be absorbed into the $SU(2)$ on the right hand side. It then suffices to prove that
$$\mathcal{T}\cdot X_\alpha\cdot \mathcal{S}=SU(2)$$
for all $\alpha$.

Following the notation in \Cref{sec:uniqueness}, we write a $Z$-rotation as $z\in U(1)$ and $X_\alpha = a+bj$, where $a^2+b^2=1$. An element in $\mathcal{T}\cdot X_\alpha\cdot \mathcal{S}$ can be parameterized as
\begin{align}
&z_3(a+bj)\overline{z_1}\frac{1+j}{\sqrt{2}}z_1\overline{z_2}\frac{1+j}{\sqrt{2}}z_2\nonumber \\
=&\frac{1}{2}(az_3+bz_3j)(1+\overline{z_1^2}j)(1+\overline{z_2^2}j)\nonumber \\
=&\frac{1}{2}\Big(a(1-u_1)-b(1+u_1)u_2)\Big)u_3\\+&\frac{1}{2}\Big(b(1-\overline{u_1})+a(1+\overline{u_1})\overline{u_2}\Big)u_3j
\end{align}
where $u_1\leftarrow \overline{z_1^2z_2^2}, u_2\leftarrow \overline{z_2^2},u_3\leftarrow z_3$.
It then suffices to prove that for arbitrary $x,y\in\mathbb{C}, |x|^2+|y|^2=1$, there exist unit complex numbers $u_1,u_2,u_3$ such that 
\begin{align}
&\begin{cases}
\frac{1}{2}\Big(a(1-u_1)-b(1+u_1)u_2)\Big)u_3=x\\
\frac{1}{2}\Big(b(1-\overline{u_1})+a(1+\overline{u_1})\overline{u_2}\Big)u_3=y
\end{cases}\\
\Leftrightarrow & 
\begin{cases}
a(1-u_1)-b(1+u_1)u_2=2x\overline{u_3}\\
b(1-u_1)+a(1+u_1)u_2=2\overline{y}u_3
\end{cases}\\
\Leftrightarrow & 
\begin{cases}
1-u_1=2ax\overline{u_3}+2b\overline{y}u_3\\
(1+u_1)u_2=2a\overline{y}u_3-2bx\overline{u_3}
\end{cases}
\end{align}
We first show that there always exists a pair $(u_1,u_3)$ satisfying the first equation, and show that $u_2$ can always be found given such a pair. Rewriting the first equation, we get
$$u_1 = 1-2ax\overline{u_3}-2b\overline{y}u_3.$$
With $u_1$ and $u_3$ running over $[0,2\pi]$, the left hand side forms a unit circle centered at the origin and the right hand side a closed symmetric curve centered at $1$. It suffices to show that the two curves intersect, i.e. there exists $u_3$ such that $$|1-2ax\overline{u_3}-2b\overline{y}u_3|=1.$$
Since $t(u_3)=1-2ax\overline{u_3}-2b\overline{y}u_3$ forms a symmetric closed curve centered at $1$, there must exist $u_3^*$ such that $t(u^*_3)\in\mathbb{R}$, hence so does $t(-u_3^*)$. By Cauchy's inequality, $$|2ax\overline{u^*_3}+2b\overline{y}u^*_3|\leq 2 (|a|^2+|b|^2)^{1/2}(|x\overline{u^*_3}|^2+|\overline{y}u^*_3|^2)^{1/2}=2,$$
and we can assume without loss of generality that $|t(u^*_3)|\leq 1\leq |t(-u^*_3)|$. Since the range of $|t(u_3)|$ must be continuous, there exists a $u'_3$ such that $|t(u'_3)|=1$.

Once we solve $u_1$ and $u_3$, we claim that there always exists $u_2$ satisfying the second constraint. It suffices to show that $|1+u_1|=|2a\overline{y}u_3+2bx\overline{u_3}|$; in fact,
$$|2ax\overline{u_3}+2b\overline{y}u_3|^2+|2a\overline{y}u_3-2bx\overline{u_3}|^2=4=|1+u_1|^2+|1-u_1|^2,$$
completing the proof.
\end{proof}

The proof of \Cref{thm:vr2} directly yields an algorithm compiling an arbitrary single-qubit gate using the vR-2 scheme; however the explicit form is rather compilated as the solution involves solving for the intersection of a unit circle with another curve (an ellipse in general), resulting in a quartic equation. 

Similarly, we prove the existence of the vR-1 scheme.

\begin{thm}
\label{thm:vr1}
Given an arbitrary rotation axis $R=aX+bY+cZ$, for arbitrary $U\in SU(2)$, there exists angles $\phi,\sigma, \theta$ such that
$$U = R_\theta X_{\sigma}(\phi).$$
Alternatively, let $\mathcal{Q}=\{X_{\sigma}(\phi)|\phi,\omega\in[0,2\pi]\}$, $\mathcal{T}_R=\{R_\theta|\theta\in[0,2\pi]\}$. Then
$$\forall R, \mathcal{T}_R\cdot \mathcal{Q}=SU(2).$$
\end{thm}
\begin{proof}
Similar to \Cref{thm:vr2}, it suffices to prove that 
$$\mathcal{T}\cdot X_\alpha\cdot\mathcal{Q}=SU(2),\forall \alpha\in[0,\pi/2].$$
Writing $X_\alpha=a+bj$ and a generic element in $\mathcal{Q}$ as $\overline{z_1}(c+dj)z_1$, every element in $\mathcal{T}\cdot X_\alpha\cdot\mathcal{Q}$ can be expressed as
\begin{align}
&z_2(a+bj)\overline{z_1}(c+dj)z_1\\
=&(az_2+bz_2j)(c+d\overline{z_1^2}j)\\
=&(au_2+bu_2j)(c+du_1j)\quad (u_1:=\overline{z_1^2}, u_2:=z_2)\\
=&(ac-bd\overline{u_1})u_2+(adu_1+bc)u_2j
\end{align}
We only need to prove that any unit pair $(x,y)$, there exists unit complex number $u_1,u_2$ and real unit pair $(c,d)$ such that 
\begin{align}
&\begin{cases}
(ac-bd\overline{u_1})u_2=x\\
(adu_1+bc)u_2=y
\end{cases}\\
\Leftrightarrow &\begin{cases}
c=a\overline{x}u_2+by\overline{u_2}\\
du_1=ay\overline{u_2}-b\overline{x}u_2
\end{cases}
\end{align}
Actually, when $u_2$ go through the unit circle, the orbit of $a\overline{x}u_2+by\overline{u_2}$ is a closed curve centric about the original point, hence it intersects the real axis. In other words, there exists unit complex number $u_2$ such that $a\overline{x}u_2+by\overline{u_2}$ is a real number, denoted by $c$. It is not hard to verify
$$|a\overline{x}u_2+by\overline{u_2}|^2+|ay\overline{u_2}-b\overline{x}u_2|^2=1$$
So there exists real number $d=\sqrt{1-c^2}$ and unit complex number $u_1$ such that $ayu_2-bx\overline{u_2}=du_1$.

\end{proof}
The proof of \Cref{thm:vr1} yields an explicit algorithm that involves solving the intersection of an ellipse and the real axis, thus only giving rise to a quadratic equation. 

\section{Compilation using virtual $R$ in the presence of leaky gates}
\label{sec:leaky}

Similar to phase carriers taking advantage of virtual $Z$ compilation schemes to reduce the single qubit gate overhead, we can use the virtual $R$ compilation schemes whenever the native two-qubit gate is leaky. Such gates include the $\CNOT$ gate and $\mathrm{DDCZ}$. In fact we prove that leaky gates are exactly local equivalents of phase carriers, as characterized in~\Cref{prp:local_eq}.  Consequently, a two-qubit gate leaking on one wire must leak on the other as well, and the virtual $R$ compilation schemes can apply to all the single-qubit gates in the circuit, with the exception of the single-qubit gates directly before a measurement. However, the single-qubit gates can be compiled with the virtual $Z$ schemes as measurements in the computational basis commutes with $Z$-rotations.

\begin{figure}[h]
    \centering
    \includegraphics[width = 0.49\textwidth]{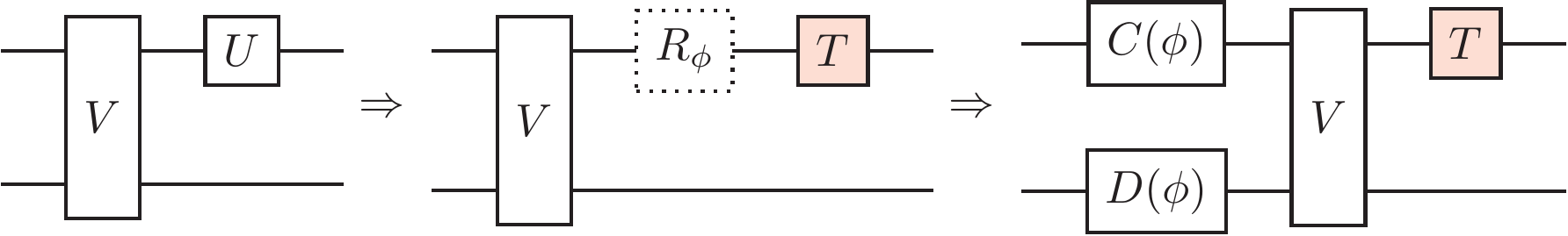}
    \caption{Illustration of the virtual $R$ compilation scheme in the presence of leaky gates. When a single-qubit unitary $U$ is followed by a leaky gate $V$, one can apply the virtual $R$ compilation schemes and decompose $U=R_\phi\cdot T$, where $T\in\mathcal{S}(\mathcal{Q})$ is to be applied physically with pulses. The extraneous rotation $R_\phi$ is then ``leaked'' through the two-qubit gate and will be jointly compiled with the following gates. Implementing $T$ saves one pulse compared to compiling $U$ with PMW compilation schemes.}
\end{figure}

\begin{thm}
A two-qubit gate $V$ is leaky on the first wire if and only if it is locally equivalent to a phase carrier. Consequently, it leaks on the other wire.
\end{thm}

\begin{proof}
We start from the condition
$$V(R_\phi\otimes I)V^\dag = C(\phi)\otimes D(\phi).$$
Without loss of generality, we assume $C(\phi), D(\phi)$ both to be differentiable with respect to $\phi$. Diffrentiating over both sides with respect to $\phi$ we get
\begin{align*}
    V(-iR\cdot R_\phi\otimes I)V^\dag&=C'(\phi)\otimes D(\phi)+C(\phi)\otimes D'(\phi)\\
    V(R\otimes I)V^\dag&=i(C'(\phi)C^\dag(\phi)\otimes I+I\otimes D'(\phi)D^\dag(\phi)).
\end{align*}
The left-hand side is a constant, and hence are $P:=iC'(\phi)C^\dag(\phi)=iC'(0)C^\dag(0)$ and $Q:=iD'(\phi)D^\dag(\phi)=iD'(0)D^\dag(0)$. The only way that the spectrums on both side match is that one of $P,Q$ is 0 and the spectrum of the other is $\{\pm 1\}$. Up to local equivalence and the $\SWAP$ gate, one can assume that $R=P=Z$ and $Q=0$, i.e.
$$V(Z\otimes I)V^\dag = Z\otimes I\Leftrightarrow [V,Z\otimes I]=0.$$
The only way that a unitary $V$ commutes with $Z\otimes I$ is that it is a controlled gate $V=|0\rangle\langle 0|\otimes U_0 + |1\rangle\langle 1|\otimes U_1$, making it locally equivalent to a CPHASE family gate. Therefore every leaky gate is locally equivalent to a CPHASE family gate up to a SWAP, exactly the local equivalence classes spanned by the phase carriers, and consequently leaks on the other wire as well. Conversely, every gate locally equivalent to a phase carrier is by definition leaky.
\end{proof}


\begin{thebibliography}{10}

\bibitem{lao2022software}
Lingling Lao, Alexander Korotkov, Zhang Jiang, Wojciech Mruczkiewicz, Thomas~E
  O'Brien, and Daniel Browne.
\newblock Software mitigation of coherent two-qubit gate errors.
\newblock {\em Quantum Science and Technology}, 2022.

\bibitem{huang2021towards}
Cupjin Huang, Dawei Ding, Feng Wu, Linghang Kong, Fang Zhang, Xiaotong Ni,
  Yaoyun Shi, Hui-hai Zhao, and Jianxin Chen.
\newblock Towards ultra-high fidelity quantum operations: Sqisw gate as a
  native two-qubit gate.
\newblock {\em arXiv preprint arXiv:2105.06074}, 2021.

\bibitem{kong2021framework}
Linghang Kong.
\newblock A framework for randomized benchmarking over compact groups.
\newblock {\em arXiv preprint arXiv:2111.10357}, 2021.

\bibitem{nielsen2010quantum}
Michael~A Nielsen and Isaac~L Chuang.
\newblock {\em Quantum Computation and Quantum Information}.
\newblock Cambridge University Press, 2010.

\bibitem{chen2016measuring}
Zijun Chen, Julian Kelly, Chris Quintana, R~Barends, B~Campbell, Yu~Chen,
  B~Chiaro, A~Dunsworth, AG~Fowler, E~Lucero, et~al.
\newblock Measuring and suppressing quantum state leakage in a superconducting
  qubit.
\newblock {\em Physical review letters}, 116(2):020501, 2016.

\bibitem{barends2014rolling}
R~Barends, J~Kelly, A~Veitia, A~Megrant, AG~Fowler, B~Campbell, Y~Chen, Z~Chen,
  B~Chiaro, A~Dunsworth, et~al.
\newblock Rolling quantum dice with a superconducting qubit.
\newblock {\em Physical Review A}, 90(3):030303, 2014.

\bibitem{krantz2019quantum}
Philip Krantz, Morten Kjaergaard, Fei Yan, Terry~P Orlando, Simon Gustavsson,
  and William~D Oliver.
\newblock A quantum engineer's guide to superconducting qubits.
\newblock {\em Applied Physics Reviews}, 6(2):021318, 2019.

\bibitem{mckay2017efficient}
David~C McKay, Christopher~J Wood, Sarah Sheldon, Jerry~M Chow, and Jay~M
  Gambetta.
\newblock Efficient z gates for quantum computing.
\newblock {\em Physical Review A}, 96(2):022330, 2017.

\bibitem{ball2016role}
Harrison Ball, William~D Oliver, and Michael~J Biercuk.
\newblock The role of master clock stability in quantum information processing.
\newblock {\em npj Quantum Information}, 2(1):1--8, 2016.

\bibitem{mi2021information}
Xiao Mi, Pedram Roushan, Chris Quintana, Salvatore Mandrà, Jeffrey Marshall,
  Charles Neill, Frank Arute, et~al.
\newblock Information scrambling in quantum circuits.
\newblock {\em Science}, 374(6574):1479--1483, 2021.

\bibitem{qiskitRZ}
Qiskit circuit library rzxgate.
\newblock Qiskit 0.23.6 documentation,
  https://qiskit.org/documentation/stubs/qiskit.circuit.library.RZ\\
  XGate.html.

\bibitem{guo2018dephasing}
Qiujiang Guo, Shi-Biao Zheng, Jianwen Wang, Chao Song, Pengfei Zhang, Kemin Li,
  Wuxin Liu, Hui Deng, Keqiang Huang, Dongning Zheng, et~al.
\newblock Dephasing-insensitive quantum information storage and processing with
  superconducting qubits.
\newblock {\em Physical review letters}, 121(13):130501, 2018.

\bibitem{foxen2020demonstrating}
Brooks Foxen, Charles Neill, Andrew Dunsworth, Pedram Roushan, Ben Chiaro,
  Anthony Megrant, Julian Kelly, Zijun Chen, Kevin Satzinger, Rami Barends,
  et~al.
\newblock Demonstrating a continuous set of two-qubit gates for near-term
  quantum algorithms.
\newblock {\em Physical Review Letters}, 125(12):120504, 2020.

\bibitem{roth2017analysis}
Marco Roth, Marc Ganzhorn, Nikolaj Moll, Stefan Filipp, Gian Salis, and
  Sebastian Schmidt.
\newblock Analysis of a parametrically driven exchange-type gate and a
  two-photon excitation gate between superconducting qubits.
\newblock {\em Physical Review A}, 96(6):062323, 2017.

\bibitem{Note1}
To disambiguate, we refer to virtual $Z$ gates or schemes as the use of PMW
  pulses to obtain an effective $Z$ rotation $\cdot Z_\phi $, while we refer to
  our proposed PMW gates or schemes as the use PMW pulses to obtain an
  effective $Z$ conjugation $Z_{-\phi } \cdot Z_\phi $. The former requires
  phase carrying while the latter does not. In other words, PMW compiled gates
  do not contribute to tracked phases.

\bibitem{Note2}
Note however that the cross-resonance gate is microwave-activated~\cite {paraoanu2006microwave,rigetti2010fully}, in which case we do not need thetwo-qubit gate to be a phase carrier.

\bibitem{sorensen2000entanglement}
Anders S{\o}rensen and Klaus M{\o}lmer.
\newblock Entanglement and quantum computation with ions in thermal motion.
\newblock {\em Physical Review A}, 62(2):022311, 2000.

\bibitem{bruzewicz2019trapped}
Colin~D Bruzewicz, John Chiaverini, Robert McConnell, and Jeremy~M Sage.
\newblock Trapped-ion quantum computing: Progress and challenges.
\newblock {\em Applied Physics Reviews}, 6(2):021314, 2019.

\bibitem{saki2021shuttle}
Abdullah~Ash Saki, Rasit~Onur Topaloglu, and Swaroop Ghosh.
\newblock Shuttle-exploiting attacks and their defenses in trapped-ion quantum
  computers.
\newblock {\em IEEE Access}, 10:2686--2699, 2021.

\bibitem{burkard2021semiconductor}
Guido Burkard, Thaddeus~D Ladd, John~M Nichol, Andrew Pan, and Jason~R Petta.
\newblock Semiconductor spin qubits.
\newblock {\em arXiv preprint arXiv:2112.08863}, 2021.

\bibitem{dobrovitski2013quantum}
VV~Dobrovitski, GD~Fuchs, AL~Falk, C~Santori, and DD~Awschalom.
\newblock Quantum control over single spins in diamond.
\newblock {\em Annu. Rev. Condens. Matter Phys.}, 4(1):23--50, 2013.

\bibitem{wang2022realizing}
Tenghui Wang, Gengyan Zhang, Hsiang-Sheng Ku, Feng Wu, Xizheng Ma, and Ran Gao.
\newblock Realizing high-fidelity and low-leakage gates in a fluxonium
  processor.
\newblock {\em Bulletin of the American Physical Society}, 2022.

\bibitem{peterson2020fixed}
Eric~C Peterson, Gavin~E Crooks, and Robert~S Smith.
\newblock Fixed-depth two-qubit circuits and the monodromy polytope.
\newblock {\em Quantum}, 4:247, 2020.

\bibitem{abrams2020implementation}
Deanna~M Abrams, Nicolas Didier, Blake~R Johnson, Marcus~P da~Silva, and Colm~A
  Ryan.
\newblock Implementation of xy entangling gates with a single calibrated pulse.
\newblock {\em Nature Electronics}, pages 1--7, 2020.

\bibitem{kelly2014optimal}
Julian Kelly, R~Barends, B~Campbell, Y~Chen, Z~Chen, B~Chiaro, A~Dunsworth,
  Austin~G Fowler, I-C Hoi, E~Jeffrey, et~al.
\newblock Optimal quantum control using randomized benchmarking.
\newblock {\em Physical review letters}, 112(24):240504, 2014.

\bibitem{xu2020experimental}
Yuan Xu, Ziyue Hua, Tao Chen, Xiaoxuan Pan, Xuegang Li, Jiaxiu Han, Weizhou
  Cai, Yuwei Ma, Haiyan Wang, YP~Song, et~al.
\newblock Experimental implementation of universal nonadiabatic geometric
  quantum gates in a superconducting circuit.
\newblock {\em Physical Review Letters}, 124(23):230503, 2020.

\bibitem{Note3}
Note that if a single-qubit gate follows we can always absorb it into the
  single-qubit gate we're compiling.

\bibitem{zhang2003geometric}
Jun Zhang, Jiri Vala, Shankar Sastry, and K~Birgitta Whaley.
\newblock Geometric theory of nonlocal two-qubit operations.
\newblock {\em Physical Review A}, 67(4):042313, 2003.

\bibitem{Note4}
Note that technically $U\in SU(4)$, so $d = -(a+b+c)$.

\bibitem{cross2019validating}
Andrew~W Cross, Lev~S Bishop, Sarah Sheldon, Paul~D Nation, and Jay~M Gambetta.
\newblock Validating quantum computers using randomized model circuits.
\newblock {\em Physical Review A}, 100(3):032328, 2019.

\bibitem{crooks2020gates}
Gavin~E Crooks.
\newblock Gates, states, and circuits.
\newblock Technical report, Berkeley Institute for Theoretical Science, 2020.

\bibitem{arute2020observation}
Frank Arute, Kunal Arya, Ryan Babbush, Dave Bacon, Joseph~C Bardin, Rami
  Barends, Andreas Bengtsson, Sergio Boixo, Michael Broughton, Bob~B Buckley,
  et~al.
\newblock Observation of separated dynamics of charge and spin in the
  fermi-hubbard model.
\newblock {\em arXiv preprint arXiv:2010.07965}, 2020.

\bibitem{paraoanu2006microwave}
GS~Paraoanu.
\newblock Microwave-induced coupling of superconducting qubits.
\newblock {\em Physical Review B}, 74(14):140504, 2006.

\bibitem{rigetti2010fully}
Chad Rigetti and Michel Devoret.
\newblock Fully microwave-tunable universal gates in superconducting qubits
  with linear couplings and fixed transition frequencies.
\newblock {\em Physical Review B}, 81(13):134507, 2010.

\end{thebibliography}

\end{document}